%% file: samplepaper.tex
\documentclass[runningheads]{llncs}
\usepackage{graphicx}
\usepackage{algorithm}
\usepackage{algorithmicx}
\usepackage{amsfonts}
\usepackage{amsmath}
\usepackage{amssymb,stmaryrd}
\usepackage{bm}
\usepackage{subfigure}
\usepackage{xcolor}         
\usepackage{float}
\usepackage{newfloat}
\DeclareFloatingEnvironment[
    fileext=los,
    listname=List of Protocol,
    name=Protocol,
    placement=H,
    within=section,
]{protocol}
\usepackage{adjustbox}
\usepackage{authblk}            
\usepackage{multirow}
\usepackage{booktabs}
\usepackage{etoolbox}
\usepackage{kantlipsum} 
\newcommand{\repthanks}[1]{\textsuperscript{\ref{#1}}}
\makeatletter
\patchcmd{\maketitle}
  {\def\thanks}
  {\let\repthanks\repthanksunskip\def\thanks}
  {}{}
\patchcmd{\@maketitle}
  {\def\thanks}
  {\let\repthanks\@gobble\def\thanks}
  {}{}
\newcommand\repthanksunskip[1]{\unskip{}}
\makeatother

\makeatletter
\newcommand{\printfnsymbol}[1]{%
  \textsuperscript{\@fnsymbol{#1}}%
}
\makeatother

\newtheorem{defx}{Definition}

\usepackage[margin=3.5cm]{geometry}

\begin{document}
\title{Secure Weighted Aggregation for Federated Learning}
%
%
\author{Jiale Guo\inst{1} \thanks{Both authors contributed equally to this research.\protect\label{X}}, Ziyao Liu\inst{1} \repthanks{X} \and
Kwok-Yan Lam\inst{1} \and
Jun Zhao\inst{1} \and
Yiqiang Chen\inst{2} \and
Chaoping Xing\inst{3}}
\authorrunning{J. Guo and Z. Liu et al.}
%
\institute{Nanyang Technological University, Singapore \\
\email{\{jiale001,ziyao002\}@e.ntu.edu.sg},
\email{\{kwokyan.lam,junzhao\}@ntu.edu.sg} \and
University of Chinese Academy of Sciences \& Peng Cheng Laboratory, Beijing, China
\email{yqchen@ict.ac.cn} \and
Shanghai Jiao Tong University, Shanghai, China\\
\email{xingcp@sjtu.edu.cn}}
\maketitle              
\begin{abstract}
The pervasive adoption of Internet-connected digital services has led to a growing concern in the personal data privacy of their customers. On the other hand, machine learning (ML) techniques have been widely adopted by digital service providers to improve operational productivity and customer satisfaction. ML inevitably accesses and processes users’ personal data, which could potentially breach the relevant privacy protection regulations if not performed carefully. The situation is exacerbated by the cloud-based implementation of digital services when user data are captured and stored in distributed locations, hence aggregation of the user data for ML could be a serious breach of privacy regulations. In this backdrop, Federated Learning (FL) is an emerging area that allows ML on distributed data without the data leaving their stored location. Typically, FL starts with an initial global model, with each datastore uses its local data to compute the gradient based on the global model, and uploads their gradients (instead of the data) to an aggregation server, at which the global model is updated and then distributed to the local data stores iteratively. However, depending on the nature of the digital services, data captured at different locations may carry different significance to the business operation, hence a weighted aggregation will be highly desirable for enhancing the quality of the FL-learned model. Furthermore, to prevent leakage of user data from the aggregated gradients, cryptographic mechanisms are needed to allow secure aggregation of FL.  Previous works on FL mainly focus on the FL framework and mechanisms for weighted aggregation. In this paper, we propose a privacy-enhanced FL scheme for supporting secure weighted aggregation. Besides, by devising a verification protocol based on Zero-Knowledge Proof (ZKP), the proposed scheme is capable of guarding against fraudulent messages from FL participants. Experimental results show that our scheme is practical and secure. Compared to existing FL approaches, our scheme achieves secure weighted aggregation with an additional security guarantee against fraudulent messages with an affordable 1.2 times runtime overheads and 1.3 times communication costs.

\end{abstract}
\input{samplebody}

%
%
%

%

\end{document}

%% file: samplebody.tex
\section{Introduction}

The pervasive adoption of Internet-connected digital services has led to a growing concern in personal data privacy of their customers. On the other hand, machine learning (ML) techniques have been widely adopted by digital service providers to improve operational productivity and customer satisfaction. ML inevitably accesses and processes users’ personal data, which could potentially breach the relevant privacy protection regulations if not performed carefully. The situation is exacerbated by the cloud-based implementation of digital services when user data are captured and stored in distributed locations, hence aggregation of the user data for ML could be a serious breach of privacy regulations. Such a scenario where multiple data stores (or custodians) jointly solve a machine learning problem while complying with privacy regulations has attracted tremendous attention from academia and industry. Privacy-preserving techniques such as differential privacy (DP), fully homomorphic encryption (FHE), and secure multi-party computation (MPC) are widely believed to be promising approaches to achieve this goal. However, it is well known that DP requires a tradeoff between data usability and privacy \cite{zhao2020local}, while MPC and HE offer cryptographic privacy with high communication or computation overheads.

In this backdrop, Federated Learning (FL) \cite{kairouz2019advances} is an emerging area that allows ML on distributed data without the data leaving their stored location. For example, an international bank typically has data centres in multiple countries, domestic regulations of each country may not allow the bank to move customer data out of the jurisdiction; hence FL will be a useful approach to allow ML overall distributed data through the participation of all local data centres. Typically, FL starts with an initial global model, with each data store (in the respective data centre) uses its local data to compute the gradient based on the global model, and uploads their gradients (instead of the data) to an aggregation server, at which the global model is updated and then distributed to the local data stores iteratively.

However, depending on the nature of digital services and FL systems, data captured at different locations may exhibit disparity. For instance, digital inclusion may vary from country to country and, due to the need for service localization, the type of collected data may also be different. Besides, the technical competence of local personnel may also vary, thus ML at local data stores may suffer from label biases. This means that different local datasets may have different significance to the global FL model, and thus the organization's business operation. Hence a weighted aggregation scheme will be highly desirable for enhancing the quality of the FL-learned model. In this case, the central server is required to evaluate the disparity across all local datasets in order to compute the weightings, and then aggregate all the clients' locally trained models according to their weights. To calculate the weight, measurements on data size and data quality are widely adopted. For example, in the most popular FL approach FedAvg \cite{mcmahan2017communication}, the locally trained models from users are weighted by the percentage of their data size in the total training data. Based on FedAvg, the authors in \cite{chen2020focus} further evaluate clients' label quality by calculating the mutual cross-entropy between data and models of both the central server and clients. Although approaches in previous works resulted in a good performance in dealing with disparity across local datasets and improved the accuracy of the global FL model, they did not address the privacy guarantee for clients' data.

Furthermore, from the angle of privacy protection, digital service providers need to address another important issue of FL, that is the feasibility of information leakage from the gradients. As pointed out in \cite{zhu2019deep}, it is possible to derive the local training data from the publicly shared gradients, or models, in a multi-node machine learning system such as an FL system. Experimental results showed that such deep leakage from gradients allows the attacker to achieve pixel-wise accurate recovery for images and token-wise matching for texts in common computer vision and natural language processing tasks. Secure aggregation schemes \cite{bonawitz2017practical,bell2020secure} aim to deal with this issue. However, during the FL process, the distributed data stores and the central server may still receive fraudulent messages due to insider frauds, which may affect the accuracy of the global FL model. For example, consider an enterprise FL system. FL participants are required to annotate the data with labels, and then train their local models with the labeled data. The competence and productivity of each FL participant will have a non-trivial impact on the label quality of the training data, which should understandably carry different weights in the overall ML. The weights are calculated by the data disparity evaluation scheme as mentioned earlier. The larger weight usually serves as an indicator of the higher quality, hence the competence and productivity of the local personnel, and vice versa. In this case, an incompetent or dishonest employee may manipulate the computed weight and send fraudulent messages (by manipulating the weights) during data disparity evaluation so as to obtain a larger weight. In other cases, the FL participant may upload a fraudulent model to affect the distribution of the global FL model in the current FL round in order to increase its weight in the next FL round \cite{blanchard2017machine}. Note that such an issue regarding fraudulent messages can be generalized to any FL system.

In order to address the aforementioned issues, we propose a secure weighted aggregation scheme to deal with both data disparity, data privacy, and dishonest participants (who send fraudulent messages to manipulate the computed weights) in FL systems. Our scheme leverages on HE and zero-knowledge proof (ZKP), and combines with a dropout-resilient secure aggregation scheme. To summarize, our contributions are the following:

\begin{itemize}
    \item By leveraging on HE, the proposed scheme achieves data disparity evaluation in a privacy-preserving manner.
    \item By devising a verification protocol based on ZKP, the proposed scheme is capable of guarding against fraudulent messages from being received by FL participants.
    \item The proposed scheme can be generalized to any data disparity evaluation scheme for secure weighted aggregation in FL.
    \item Compared to existing FL schemes, experimental results show that our scheme achieves secure weighted aggregation with an additional security guarantee against fraudulent messages with an affordable 1.2 times of runtime overheads and 1.3 times communication costs.
\end{itemize}

\textbf{Organisation of the paper:} The rest of the paper is organized as follows. In Section~\ref{sec:preliminaries}, we introduce the preliminaries and supporting protocols. In Section~\ref{sec:technical-intuition}, we describe the rationales for secure weighted aggregation together with a high-level overview, and the threat model. Then we give a detailed description of our proposed scheme, and discuss its privacy and security in Section~\ref{sec:our-protocol}. Experimental results are present in Section~\ref{sec:performance}. Finally, we discuss related works in Section~\ref{sec:related-works}, and conclude and outline several directions for future research in Section~\ref{sec:conclusion}. More background, detailed proofs and tolerance analysis are given in the Appendix.

\section{Preliminaries and Supporting Protocols} 
\label{sec:preliminaries}
This section briefly describes the preliminaries of federated learning, data disparity evaluation, secure aggregation, and related cryptographic protocols used in this work.

\subsection{Federated Learning}
\label{sed:federated-learning}
Federated learning scheme enables multiple users\footnote{We use the terms user and client interchangeably.} to jointly solve a machine learning problem with the coordination of a central server. The participants of FL are divided into two classes: a server $S$ and $K$ users $\mathcal{U}=\{u_1,u_2,\dots,u_K\}$ that each user $u_i\in \mathcal{U}$ holds a local dataset $\mathcal{D}_i=\{(x_j^i,y_j^i)|j=1,\dots,n_i\}$ with size $n_i$, where $x_j^i$ and $y_j^i$ are the $j$-th sample and corresponding label in $\mathcal{D}_i$ respectively. The size of the whole $K$ training data of the user is $n=\sum_{i=1}^K n_i$.

When the participants collaboratively training an FL model $\mathcal{M}$, the training objective can be represented as  $\arg\min\limits_{\mathcal{M}\in \mathbb{R}^d} f(\mathcal{M})$, where 
$$
f(\mathcal{M})=\sum_{i=1}^{K} \frac{n_{i}}{n} F_{i}(\mathcal{M}),\ F_{i}(\mathcal{M})=\frac{1}{n_{i}} \sum_{(\mathbf{x}_j,y_j) \in \mathcal{D}_{i}} f_k(x_j,y_j;\mathcal{M})$$
Here, $f_k(x_j,y_j;\mathcal{M})$ computes the loss of the model $\mathcal{M}$'s prediction on $j$-th example $(x_j,y_j)$.

A typical FL scheme works by repeating the following steps until training is stopped, determined by the model engineer monitoring the training process.

\begin{enumerate}
    \item User selection: In round $t$, the server $S$ selects a fraction of users $\mathcal{U}_t\subseteq \mathcal{U}$.
    \item Broadcast: The server $S$ broadcasts the current global model $M_t$ to the selected users.
    \item User computation: Each selected user $u_i\in \mathcal{U}_t$ trains a local model on their dataset $D_i$, and calculates the update $M_t^i$. 
    \item Model aggregation: The server $S$ collects local updates from users in the current round and computes global update according to an aggregation method, e.g., in \cite{mcmahan2017communication}, the local update of user is weighted by the percentage of their data size in the whole training data.
    \item Model update: The server $S$ updates the global model such that $\mathcal{M}_{t+1}\leftarrow \mathcal{M}_t^s$.
\end{enumerate}

\subsection{Logistic Regression}
\label{sec:logistic-regression}

Standard logistic regression is a linear algorithm with a non-linear transform on output for binary classification. Denote the input sample $\bm{x}$ as an $n$ dimensional vector $\bm{x}=(x_1, \dots, x_n)$, and the model parameter $\bm{\theta}$ as an $n+1$ dimensional vector $\bm{\theta}=(\theta_0, \theta_1, \dots, \theta_n)$. The output results $y$ which lies between 0 and 1 can be obtained by computing:
\begin{equation*}
    y=s(\bm{\theta \cdot x})=s(h_{\theta}(\bm{x}))=\frac{1}{1+e^{-\bm{\theta \cdot x}}}
\end{equation*}
where $h_{\theta}(\bm{x}) = \bm{\theta \cdot x}=\theta_0+\sum_{k=1}^n\theta_kx_k$ and $s(x)=\frac{1}{1+e^{-x}}$ is the sigmod function which is widely adopted in machine learning. 

The objective of training a logistic regression model using dataset $\mathcal{D}=\{(x_k,y_k)|k=1,\dots,|\mathcal{D}|\}$ is to optimize the following cost function:
$$
J(\bm{\theta})=\frac{1}{m}\left[\sum_{k=1}^{m}-y_{k} \log \left(h_{\theta}\left(\bm{x_{k}}\right)\right)+\left(1-y_{k}\right) \log \left(1-h_{\theta}\left(\bm{x_{k}}\right)\right)\right]$$
where $m=|\mathcal{D}|$ is the number of samples in $\mathcal{D}$. For round $i$, an optimal $\bm{\theta}$ in round $i+1$ can be computed using gradient descent algorithm by:
$$
\bm{\theta}^{i+1} \leftarrow \bm{\theta}^{i}-\eta \nabla J\left(\bm{\theta}^{i}\right), i \geq 0$$
where $\eta$ is the learning rate and $\nabla J$ is the gradient of $J$. The iterative optimization on $\bm{\theta}$ stops determined by the model engineer. 

\subsection{Data Disparity Evaluation}
\label{sec:data-disparity}
As data size and data quality vary across local datasets, an evaluation scheme for data disparity is essential in the FL system. In this work, we improve the label quality evaluation scheme proposed in \cite{chen2020focus} to illustrate how to construct an interactive evaluation scheme in a privacy-preserving manner. Note that this scheme can be generalized to any data disparity evaluation scheme for federated learning.

By maintaining a small set of benchmark dataset $\mathcal{D}_s$, the central server $S$ is allowed to quantify the credibility $C_{i}$ of each local dataset $\mathcal{D}_i$ by computing the mutual cross-entropy $E_{i}$. In specific, mutual cross-entropy $E_{i}$ evaluates both the performance of the global model $\mathcal{M}$ on the local dataset $\mathcal{D}_i$, i.e., $LL_{i}$, and the performance of the local model of the user $\mathcal{M}_i$ on the benchmark dataset $\mathcal{D}_s$, i.e., $LS_{i}$, which is given by:

\begin{equation*}
\begin{aligned}
E_{i} &=LS_{i}+LL_{i} \\
LS_{i} &=-\sum_{(x_s, y_s) \in \mathcal{D}_{s}} y_s \log P\left(y \mid x_s ; \mathcal{M}_{i}\right) \\
LL_{i} &=-\sum_{(x_u, y_u) \in \mathcal{D}_{i}} y_u \log P\left(y \mid x_u ; \mathcal{M}\right)
\end{aligned}
\end{equation*}

Then the weight of user $i$'s model can be defined as: 
$$
w_i=\frac{n_{i} C_i}{\sum_{j=1}^{K} n_{j} C_i} \text{ where } C_i=1-\frac{e^{\alpha E_i}}{\sum_{j=1}^{K} e^{\alpha E_j}}$$

Here $n_i$ is the size of the dataset $\mathcal{D}_i$, $K$ is the number of users, and $\alpha$ is a hyper-parameter for normalization. These weights can be used in subsequent weighted aggregation to obtain the new global model as $\sum_{i=1}^{K} w_i\mathcal{M}_{i}$.

\subsection{Secure Aggregation}
\label{sec:secure-aggregation}

Secure aggregation scheme \cite{bonawitz2017practical} enables a server $S$ to compute a sum $z$ of the locally trained models $\{\mathcal{M}_u\}_{u\in \mathcal{U}}$ from a set of users $\mathcal{U}$ in a private and dropout-robust manner, when the number of non-dropout users is greater than a threshold $t$. The improved version \cite{bell2020secure} further reduce the communication overheads by having each user communicates across only a subset of users rather than all the users. We denote such aggregation schemes as $ z \leftarrow \pi_{SecAgg}(\mathcal{U},\{\mathcal{M}_i\},t)$. Specifically, the illustrative scheme can be summarized as follows:

\textbf{Mask Generation $(\{R\}, \mathcal{U}_1) \leftarrow \pi_{MG}(t, \mathcal{U})$}: Each user $u\in \mathcal{U}$ generates a random matrix $R$ to mask the local model $\mathcal{M}$, and the set of alive users after mask generation is denoted as $\mathcal{U}_1$. In specifc, the matrix $R$ is generated using pseudorandom generator (PRG) \cite{blum1984generate} that $R=R_b+R_s$ which consists of (i) a self-chosen matrix $R_b=\text{PRG}(b_u)$ where $b_u$ is a self-chosen seed by the user $u$, and (ii) the sum of pairwise matrices with all other users $R_s=\text{PRG}(\sum_{u<v} s_{u, v}-\sum_{u>v} s_{v, u})$  (assume a total order on users) where $s_{u, v}$ is the seed that both user $u$ and $v$ agree on. To handle the dropout users, the self-chosen seed $b_u$ and the pairwise seed $s_{u, v}$, are secretly shared among all users using a $t$-out-of-$|\mathcal{U}|$ Shamir threshold secret sharing scheme.

\textbf{Masked Model Aggregation $(y, \mathcal{U}_2) \leftarrow \pi_{MMA}(\{\mathcal{M}\}$, $\{R\}$, $\mathcal{U}_1), t$}. Each user $u \in \mathcal{U}_1$ computes and uploads its masked local model $y_u=\mathcal{M} +R$. Then the server collects all masked models from the set of alive users, denoted as $\mathcal{U}_2\subseteq \mathcal{U}_1$, and computes the aggregation $y=\sum_{u\in \mathcal{U}_2}y_u$. Note that the number of alive users $|\mathcal{U}_2|$ should be larger than the threshold $t$, otherwise abort.

\textbf{Model Aggregation Recovery $z \leftarrow \pi_{MAR}(y, \mathcal{U}_1,\mathcal{U}_2,t)$}. Each user $u \in \mathcal{U}_2$ sends its Shamir shares of $b_u$ and each user $u \in \mathcal{U}_1$ sends its Shamir shares of $s_{u,v}$ to the server. After receiving the shares of $b_u$ and $s_{u,v}$ of which the number is greater than the threshold $t$, the server reconstructs corresponding $b_u$ and $s_{u,v}$ to recover the final aggregated result $z$ as:

$$
z=y-\sum_{u\in \mathcal{U}_2}\mathsf{PRG}(b_u)+\sum_{u\in \mathcal{U}_1/\mathcal{U}_2,v\in \mathcal{U}_2}\mathsf{PRG}(s_{u,v})$$

Note that for simplicity, we omit some cryptographic primitives such as the signature scheme, Diffie-Hellman key agreement scheme, and describe only the illustrative secure aggregation scheme for semi-honest setting. We refer interested readers to \cite{bonawitz2017practical,bell2020secure} for the details.

\subsection{Cryptographic Tools}

We now introduce the cryptographic building blocks used in our scheme.
\subsubsection{Paillier Crypto-system}
\label{sec:Paillier}
Paillier Crypto-system is a probabilistic asymmetric encryption scheme with additive homomorphic properties \cite{paillier1999public} consisting of a tuple of algorithms as follows:
\begin{itemize}
    \item $\mathsf{HE.KeyGen}(p,q)$: For large primes $p$ and $q$ that $gcd(pq,(p-1)(q-1)=1$, compute $n=p\cdot q$ and $\lambda=lcm(p-1,q-1)$. Then, select a random integer $g\in \mathbb{Z}_{n^2}^*$ that ensure $gcd(n, L(g^{\lambda} \bmod n^{2}))=1$, where function $L$ is defined as $L(x)=\frac{x-1}{n}$ and $\lambda=\varphi(n)$. Finally, output a key pair $(sk, pk)$ where the public key is $pk=(n,g)$ and the private key is $sk=(p,q)$ . 
    \item $\mathsf{HE.Enc}(pk,m,r)$: For message $m\in \mathbb{Z}_N$, taking the public key $pk$ and a random number $r\in \mathbb{Z}_N^*$, output a ciphertext $c=\mathsf{Enc}(m)=g^{m} r^{n} \quad(\bmod\ n^{2})$. 
    \item $\mathsf{HE.Dec}(sk, c)$: For a ciphertext $c < n^2$, taking private key $sk$, output the plaintext message $m=\mathsf{Dec}(c)=\frac{L\left(c^{\lambda}\left(\bmod n^{2}\right)\right)}{L\left(g^{\lambda}\left(\bmod n^{2}\right)\right)}$. 
\end{itemize}

For simplicity, we sometimes abuse the notation $Enc$ and $Dec$ instead of $\mathsf{HE.Enc}(pk,m,$ $r)$ and $\mathsf{HE.Dec}(sk, c)$ when the context is clear. Paillier supports homomorphic addition between two ciphertext such that $\mathsf{Enc}(m_1)\cdot \mathsf{Enc}(m_2)(\bmod\ n^{2})=\mathsf{Enc}(m_1+m_2$ $(\bmod\ n^{2}))$, and homomorphic multiplication between a plaintext and a ciphertext such that $(\mathsf{Enc}(m_{1}))^{m_{2}}$ $(\bmod\ {n^{2}})=\mathsf{Enc}(m_{1} m_{2}\ (\bmod\ n))$. We denote $\boxplus$ with the homomorphic addition between two ciphertexts, and $\boxtimes$ with the homomorphic multiplication between a ciphertext and a plaintext. In addition, Paillier crypto-system provides semantic security against chosen-plaintext attacks (IND-CPA) (see Appendix \ref{app:ind-cpa}) such that an adversary will be unable to distinguish pairs of ciphertexts based on the message they encrypt \cite{paillier1999public}.

\subsubsection{Zero Knowledge Proof of Plaintext Knowledge.}
\label{sec:zkpopk}
A zero-knowledge proof of plaintext knowledge (ZKPoPK) algorithm enables a prover to prove the knowledge of a plaintext $m$ of some ciphertext $C = Enc(m)$ to a verifier in a given public encryption scheme, without revealing anything about message $m$ \cite{goldreich1994definitions}. Following the Paillier-based ZKPoPK algorithm proposed in \cite{damgaard2010generalization}, the prover with a key pair $(pk,sk)$ can prove the knowledge of zero of a ciphertext $u=\mathsf{HE.Enc}(pk,0,v)\in \mathbb{Z}_{n^2}$ to the honest verifier. Algorithm \ref{alg:p-zkpopk} provides details of the Paillier-based ZKPoPK of Zero algorithm $b\leftarrow \mathsf{PZKPoPKoZ}(n,u,pk,sk,v)$. The correctness and security proof are given in \cite{damgaard2010generalization}.

\begin{algorithm}[htb] 
\begin{flushleft}
    \textbf{Public input:} $n,u,pk$.\\
    \textbf{Private input of the prover:} $sk$, $v\in \mathbb{Z}_n^*$, such that $u=\mathsf{HE.Enc}(pk,0,v)$.\\
    \textbf{Output:} $\beta$
\end{flushleft} 
\begin{algorithmic}[1]
\caption{Paillier-based ZKPoPK of Zero $b\leftarrow \mathsf{PZKPoPKoZ}(n,u,pk,sk,v)$}
\State The prover randomly chooses $r\in \mathbb{Z}_n^*$ and sends $a=\mathsf{HE.Enc}(pk,0,r)$ to the verifier.
\State The verifier chooses a random $e$ and sends it to the prover.
\State The prover sends $z=rv^e \bmod\ n$ to the verifier. 
\State The verifier checks that $u,a,z$ are prime to $n$ and that $\mathsf{HE.Enc}(pk,0,z)=au^e \bmod\ n^2$. If and only if these requirements satisfy, output $\beta=1$; otherwise, output $\beta=0$.
\label{alg:p-zkpopk}
\end{algorithmic}
\end{algorithm}

\section{Security and Privacy Requirements}
\label{sec:technical-intuition}

In general, the requirements of the proposed protocol are (i) both the client's data and locally trained models cannot be learned by any other party, and (ii) the server achieves its training objective using FL while the weights indicate the real competence and productivity of clients. In particular, we divide our proposed scheme into two phases (i) data disparity evaluation where the server and each user $u_i$ jointly calculate the weight $w_i$ based on the encrypted user's local dataset $\mathcal{D}_i$ and model $\mathcal{M}_i$, and (ii) weighted aggregation where the server aggregates the model of each user $\mathcal{M}_i$ according to its weights $w_i$.  The scheme is dropout-resilient that allows at most $n-t$ dropout clients where $t$ is the threshold of the Shamir secret sharing scheme.

\textbf{Threat model.} Particularly, our protocol is secure against a malicious adversary controlling the central server and up to $\lceil  \frac{n}{3} \rceil +1$ clients with the threshold $t>\lfloor  \frac{2n}{3} \rfloor+1$ for the Shamir secret sharing scheme (because these bound are derived from secure aggregation scheme \cite{bonawitz2017practical,bell2020secure} that we adapt to our proposed scheme, we refer interested readers to Appendix \ref{app:disc-tol} and \cite{bonawitz2017practical} for the details). Both the clients and the central server can send fraudulent messages to each other. Note that we assume authenticated channels and the signature schemes to ensure the confidentiality and integrity of messages sent by participants, i.e., no Sybil attack (see the system setting in Section \ref{sec:performance} for the details).

We consider only three types of active malicious behaviors, including (i) the user may publish the fraudulent weight (as we use HE for privately computing the weight, thus decryption is on the user-side, we will explain this in detail later), (ii) the user may upload the fraudulent locally trained model, and (iii) the server may give a different view of dropped users to the honest users during aggregation ( these operations are marked in underlined red in Protocol \ref{fig:our_protocol}). The reason for such consideration is based on the following natures:

\begin{itemize}
    \item The server is the party that intends to achieve the training objective using FL, and thus it would not deviate from the protocol before obtaining the real weights and models of users.
    \item The user will upload the encryption of real local dataset $\mathcal{D}_i$ and model $\mathcal{M}_i$ for data disparity evaluation. This is because the user does not know how to modify them to obtain a larger weight value.
    \item The user may publish a fraudulent weight with a larger value than the real weight for its benefit.
    \item After getting the real weight, the user may upload a fraudulent model to affect the distribution of the global FL model in the current FL round to increase its weight in the next FL round \cite{blanchard2017machine}.
    \item Since the list of dropped users is maintained by the server, the server may give a different view of dropped users to the honest users for malicious purposes, e.g., the server can claim an honest alive user to be dropped to ask other users to send its shares for recovering that user's model or gradient.
\end{itemize}

\textbf{High-level overview.} Before proceeding to the phase for data disparity evaluation, each user $u_i$ uploads its encrypted dataset $\mathcal{D}_i$ to the server $S$. Note that those encrypted datasets are uploaded only once at the very beginning of FL. After that, in the phase for data disparity evaluation, the protocol works as follows: (i) each user $u_i$ uses its local dataset $\mathcal{D}_i$ to train the model $\mathcal{M}_i$, and send encrypted model $\mathcal{M}_i$ to the server. Then (ii) the server $S$ calculates the weight $w_i$ following data disparity evaluation algorithm for each user $u_i$ using HE operations based on the received encrypted dataset $\mathcal{D}_i$ and encrypted model $\mathcal{M}_i$ from each user $u_i$. After that, (iii) the server $S$ sends encrypted weight $w_i$ to the user $u_i$, and finally (iv) the user decrypts the weight and publishes it to the server and other users. Note that the user holds the private key, and the decryption of weight is on the user-side. Thus, according to the threat model we discussed earlier, the server would not send fraudulent encrypted weight to each user, but the user may perform fraudulent decryption of the weight. Therefore, only the correctness of the decryption of weight on the user-side needs to be guaranteed. We adopt a variant of the ZKPoPK scheme to achieve such a goal (see Section \ref{sec:zkpopk}). Besides, we notice that most of the HE operations for data disparity evaluation on the server-side are only with a ciphertext, i.e., encrypted data or model of the user, and a plaintext, i.e., server's data or model. Thus, we can rely on the Paillier cryptosystem, which supports such HE operations instead of the less efficient Fully HE crypto-system (see Section \ref{sec:Paillier} and Section \ref{sec:ppdde}).

In the second phase for weighted aggregation, a secure aggregation scheme \cite{bonawitz2017practical,bell2020secure} provides a privacy guarantee against semi-honest and active adversaries. Note that the consistency check is adopted to deal with the issue regarding the central server assigning different views to honest clients, as mentioned earlier. Since we fully leverage this technique which is identical to the one used in the secure aggregation scheme, we describe it in Appendix \ref{app:con-check} and discuss the security of our proposed method in combination with consistency check, i.e., the maximum number of malicious adversaries allowed and the lower bound of the threshold $t$ for Shamir scheme, in Appendix \ref{app:disc-tol}. We follow the secure aggregation scheme in our work because it can be simply extended to secure weighted aggregation by uploading weighted models. However, an issue exists that users may upload fraudulent weighted models to the server for their benefits. For example, the clients can upload well-designed fraudulent models to affect the distribution of the global model of the current round to increase their weights in the next round \cite{blanchard2017machine}. Addressing such issues is not trivial as the solution may involve a complicated incentive mechanism and several protocols for security and privacy purposes. Luckily, since the real users' models in ciphertext have already been uploaded to the server in the first phase, and the weight of each user are public, relying on the ZKPoPK algorithm, the server can require each user to prove the plaintext knowledge of its weighted model without conveying any information apart from the fact that whether the user uploads the real weighted model or not. We note that following the secure aggregation scheme, each weighted model $w_i \mathcal{M}_i$ to be uploaded is masked by adding some randomness $R_i$. In this case, each user needs to prove the plaintext knowledge of its masked weighted model $w_i \mathcal{M}_i+R_i$ instead of the weighted model $w_i \mathcal{M}_i$, which requires the server to hold the ciphertext of the corresponding real masked weighted model. Since the encrypted masked weighted model can be obtained by adding the ciphertext of $R_i$ to the ciphertext of the real weighted model, the problem boils down to ensuring the user honestly uploading the ciphertext of $R_i$ to the server. In our scheme, we require users to upload these ciphertexts at the very beginning of each round of FL, such that users would not deviate from the protocol as they do not know how to design $R_i$ to increase their weights for benefits without knowing the information of their weights and locally trained models in the current round.

\section{Secure Weighted Aggregation Protocol}
\label{sec:our-protocol}
This section describes technical details of our privacy preserving scheme consisting for data disparity evaluation and weighted aggregation for federated learning.
\subsection{Data Disparity Evaluation}
\label{sec:ppdde}

For the phase of data disparity evaluation, we adopt a logistic regression model (see Section \ref{sec:logistic-regression}) and adapt the scheme proposed in \cite{chen2020focus} which enables the server to quantify the credibility of each user’s local dataset based on its label quality in a secure and privacy-preserving manner. 
Recall that the server $S$ maintains a benchmark dataset $\mathcal{D}_s=\{(x_s,y_s)|s=1,2,\dots,n_s\}$ and a global FL model $\mathcal{M}$, and the user $u_i$ holds a local dataset $\mathcal{D}_i=\{(x_u,y_u)|u=1,2,\dots,n_i\}$ and a locally trained model $\mathcal{M}_i$ as well as a Paillier key pair $(sk_i,pk_i)$. 

The privacy-preserving data disparity evaluation works as follows (see Section \ref{sec:data-disparity} for the version over plaintext): (i) each user $u_i$ uses its $pk_i$ to encrypt its dataset and locally trained model as $\mathsf{Enc}(\mathcal{D}_i)=\{(\mathsf{Enc}(x_u),\mathsf{Enc}(y_u))|$ $u=1,2,\dots,n_i\}$ and $\mathsf{Enc}(\mathcal{M}_i)$ respectively, and uploads them to the server. (ii) Then the server computes and sends the mutual cross entropy $E_{i}$ in ciphertext $\mathsf{Enc}(E_{i})$ to the user $i$ as:
\begin{equation}
\label{equ:entropy}
\begin{aligned}
\mathsf{Enc}(E_{i}) = \mathsf{Enc}(LS_{i})\boxplus \mathsf{Enc}(LL_{i})
\end{aligned}
\end{equation}
\begin{equation}
\label{equ:entropy_LS}
\begin{aligned}
\mathsf{Enc}(LS_{i}) = -\sum_{(x_s, y_s) \in \mathcal{D}_{s}} y_s \boxtimes \log P\left(y \mid x_s ; \mathsf{Enc}(\mathcal{M}^{i})\right)
\end{aligned}
\end{equation}
\begin{equation}
\label{equ:entropy_LL}
\begin{aligned}
\mathsf{Enc}(LL_{i}) = -\sum_{(x_u, y_u) \in \mathcal{D}_{i}} \mathsf{Enc}(y_u)\boxtimes \log P\left(y \mid \mathsf{Enc}(x_u) ; \mathcal{M}^{s}\right)
\end{aligned}
\end{equation}

As illustrated in Section \ref{sec:logistic-regression} and Section \ref{sec:data-disparity}, computation of $LS_{i}$ and $LL_{i}$ involves a sigmoid function $s(x)=\frac{1}{1+e^{-x}}$ and a logarithmic function $\text{log}(x)$ which cannot be directly evaluated using HE operations. A common method to address this issue is to replace these non-linear functions with polynomials. For instance, in terms of a good tradeoff between efficiency and accuracy, the sigmoid function can be approximated by a cubic polynomial $\sigma(x)=s_0+s_1x+s_2x^2+s_3x^3$, and thus the $y$ in Equation \ref{equ:entropy_LS} and \ref{equ:entropy_LL} can be calculated by $y=\sigma(l)=s_0+s_1l+s_2l^2+s_3l^3$ where $l=\bm{\theta \cdot x}$. Since in our scheme, either one of $\bm{\theta}$ and $\bm{x}$ is in ciphertext, the server can only obtain $\mathsf{Enc}(l)$, which means that the $n$-th power of $l$ cannot be calculated using Paillier as it does not support multiplication of ciphertext. Therefore, one more round of communication is needed to compute $y$. Specifically, the server chooses a random value $r \in 2^\kappa$ where $\kappa$ is the security parameter, to mask $l$ as $z=l+r$ and sends $\mathsf{Enc}(z)=\mathsf{Enc}(l) \boxplus \mathsf{Enc}(r)$ to the user. Then the user decrypts $\mathsf{Enc}(z)$, computes $\mathsf{Enc}(z^2)$ and $\mathsf{Enc}(\sigma(z))$ from $z$ and sends $(\mathsf{Enc}(z^2), \mathsf{Enc}(\sigma(z)))$ to the server.
Since 
\begin{equation*}
\begin{aligned}
    y=\sigma(l)=\sigma(z-r)=\sigma(z)-\sigma(r)+\left(s_{0}+3 s_{3} r^{3}\right)-3 s_{3} r z^{2}-\left(2 s_{2} r-3 s_{3} r^{2}\right) l
\end{aligned}
\end{equation*}
the $\mathsf{Enc}(y)$ can be calculated by 
\begin{equation}
\label{equ:sigmoid}
\begin{aligned}
    \mathsf{Enc}(y)=&\mathsf{Enc}(\sigma(z)) \boxplus \mathsf{Enc}(-r) \boxplus \mathsf{Enc}(s_0+3s_3r^3)\\ &\boxplus ((-3s_3r)\mathsf{Enc}(z^2)) \boxplus ((-2s_2r+3s_3r^2)\mathsf{Enc}(l))
\end{aligned}
\end{equation}
Note that since $l=\bm{\theta \cdot x}$ is masked by $r$, the user has no information of how to design fraudulent $(\mathsf{Enc}(z^2), \mathsf{Enc}(\sigma(z)))$ for its benefit. In addition, there is a difference between the computation of $\mathsf{Enc}(LS_{i})$ and $\mathsf{Enc}(LL_{i})$ as the former involves multiplication of plaintext while the latter involves multiplication of ciphertext. This means that Equation \ref{equ:entropy_LS} can be directly evaluated using Paillier and Equation \ref{equ:entropy_LL} needs one more round for computation. Specifically, let $\mathsf{Enc}(h) = \log P\left(y \mid \mathsf{Enc}(x_u) ; \mathcal{M}^{s}\right)$ in Equation \ref{equ:entropy_LL}, the server first sends $\mathsf{Enc}(h)$ to the user for decryption. To protect $h$ which may leak information of $x$ against the server, the user decrypts $\mathsf{Enc}(h)$ and masks $h$ using a randomly chosen value $r$ then sends $h+r$ to the server (here the user may send fraudulent $h+r$). Following that, the server computes $\mathsf{Enc}(y_uh+y_ur)=\mathsf{Enc}(y_u)(h+r)$. Assume that $\mathsf{Enc}(y_ur)$ has been uploaded to the server, the server is able to compute $\mathsf{Enc}(y_uh)=\mathsf{Enc}(y_uh+y_ur) \boxplus \mathsf{Enc}(-y_ur)$. Note that $r$ is secure even if the server holds $\mathsf{Enc}(y_u)$ and $\mathsf{Enc}(y_ur)$. Such security is provided by the IND-CPA property of Paillier (see Section \ref{sec:Paillier} and Appendix \ref{app:ind-cpa}). In addition, for computing binary cross-entropy after sigmod function, we can further simplify the polynomial replacement of their combined function using the similar above mentioned method.

After that, (iii) each user $i$ uses its private key $sk_i$ to decrypt $\mathsf{Enc}(E_{i})$ and publishes $E_{i}$. Finally (iv) all participants can compute the credibility $C_i$ and weight $w_i$ for each user $i$ following:
\begin{equation}
\label{equ:weight}
\begin{aligned}
w_{i} =\frac{n_{i} C_i}{\sum_{k=1}^{K} n_{k} C_i} = \frac{n_{i} e^{\alpha{}RE_i}}{\sum_{k} n_{k} e^{\alpha{}RE_k}} \text{, such that} \sum_{i} w_i=1
\end{aligned}
\end{equation}
where
$$C_i=\frac{e^{\alpha RE_i}}{\sum_k e^{\alpha RE_k}},\ RE_i=\frac{1}{E_i}$$
Here, Equation \ref{equ:weight} is the adapted version of that in \cite{chen2020focus} in order to handle dropout case by simply adjusting the denominator of $w_i$ in Equation \ref{equ:weight} on server-side during subsequent weighted aggregation (see Section \ref{sec:weighted_aggregation}).

As discussed in Section \ref{sec:technical-intuition}, since the user $i$ may deviate from the protocol for its benefit by publishing fraudulent decryption of $E_i$ in step (iii) especially when there is an incentive mechanism where reward provided by the server is distributed according to clients' weights, a verification scheme is needed to guarantee the correct decryption of $E_i$ on the user-side. To address this issue, we construct a variant of the ZKPoPK of Zero algorithm given in Section \ref{sec:zkpopk} to enable a user $u_i$ with a key pair $(pk_i,sk_i)$ to convince the server $S$ the fact that a message $m$ is the correct decryption of ciphertext $c$. To do so, given the message $m$, the server can compute $c'=c-\mathsf{Enc}(m)$ and then requires the user to prove the knowledge of zero of the ciphertext $c'$ following Algorithm \ref{alg:p-zkpopk}. The detail of this algorithm is given in Algorithm \ref{alg:p-zkp-m}. 

Assume that the server receives a message $E_i^{\prime}$ which is the decryption of $\mathsf{Enc}(E_i)$ from the user $u_i$. Following Algorithm \ref{alg:p-zkp-m}, the server $S$ and the user $u_i$ get $\beta\leftarrow \mathsf{PPoPK}(\mathsf{Enc}(E_i),$ $E_i^{\prime},$ $pk_i,sk_i)$. The server accepts $E_i^{\prime}$ as the plaintext of $\mathsf{Enc}(E_i)$ if and only if $\beta = 1$, which means the user $u_i$ honestly decrypts $\mathsf{Enc}(E_i)$. Otherwise, the server refuses the message $E_i^{\prime}$ and removes the user $u_i$ from the user set for aggregation. Similar approach can be adopted to guarantee the correct decryption of $h$ in step (ii). The correctness and security proof of the algorithm \ref{alg:p-zkpopk} can be found in Appendix \ref{app:p-zkpopk}.
\vspace{-0.4cm}
\begin{algorithm}[htb] 
\caption{Paillier-based PoPK $\beta\leftarrow \mathsf{PPoPK}(c,m,pk_i,sk_i)$}
\label{alg:p-zkp-m}
\begin{flushleft}
    \textbf{Public input:} $c,m,pk_i=(n,g)$.\\
    \textbf{Private input for the user $u_i$:} $sk_i=(p,q)$\\
    \textbf{Output:} $\beta$
\end{flushleft} 
\begin{algorithmic}[1]
\vspace{-0.2cm}
\State The server randomly chooses $r_s\in \mathbb{Z}_n^*$, and computes $c'=c-\mathsf{HE.Enc}(pk_i,m,r_s)(\bmod n^2)$. The server sends $c'$ to the user.\label{step:c'}
\State The user computes $r'=c'{}^{d}\bmod n$, where $d=n^{-1}\bmod \varphi(n)$. \label{step:r'}
\State The server and the user jointly follow Algorithm \ref{alg:p-zkpopk} to get $\beta\leftarrow \mathsf{PZKPoPKoZ}(n,c',pk_i,sk_i,r')$.
\end{algorithmic}
\end{algorithm}

\subsection{Weighted Aggregation}
\label{sec:weighted_aggregation}
\vspace{-0.1cm}
After each weight of the user is determined after the phase of data disparity evaluation, the server collects all locally trained models from users and aggregates them to update the global model $\mathcal{M}$ according to their corresponding weights:
$$
\setlength{\abovedisplayskip}{3pt}
\setlength{\belowdisplayskip}{3pt}
\mathcal{M}=\sum_{u_i\in \mathcal{U}^{\prime}}w_i\mathcal{M}_i,\text{ such that} \sum_{u_i\in \mathcal{U'}}w_i=1$$
where $w_i$ and $\mathcal{M}_i$ are the weight and local model of user $u_i$, and $\mathcal{U}^{\prime}$ is the selected user set in the current round of FL.

Let the numerator of $w_i$ in Equation \ref{equ:weight} be $\omega_i = n_i e^{\alpha RE_i}$, we can rewrite Equation \ref{equ:weight}:
\begin{equation}
\setlength{\abovedisplayskip}{3pt}
\setlength{\belowdisplayskip}{3pt}
\label{equ:omega}
    \mathcal{M}=\frac{1}{\sum_{u_j\in \mathcal{U'}} n_j e^{\alpha RE_j}}\sum_{u_i\in \mathcal{U}^{\prime}} \omega_i \mathcal{M}_i
\end{equation}

Therefore, the server can deal with dropout users by adjusting $\mathcal{U}^{\prime}$ and aggregates all alive models from users following the secure aggregation scheme \cite{bell2020secure,bonawitz2017practical}.


To prevent users from uploading fraudulent weighted models to the server, as we discussed in Section \ref{sec:technical-intuition}, a ZKPoPK based verification scheme is attached. The main idea of such a verification scheme is that if the server holds the encryption of each real masked weighted model $cy_{i}$, it can verify if the user uploads its real masked weighted model $y_i$ by checking the equality between $y_i$ and the decryption of $cy_i$. Specifically, relying on the verification Algorithm \ref{alg:p-zkp-m}, each user needs to prove the plaintext knowledge of its masked weighted model $\omega_i \mathcal{M}_i+R_i$ instead of the weighted model $\omega_i \mathcal{M}_i$, which requires the server to hold the ciphertext of the corresponding real masked weighted model. This means that each user $u_i\in \mathcal{U}^{\prime}$ has to upload its encrypted mask $\mathsf{Enc}(R_i)$ to the server, and the server computes the real encryption of each masked weighted model $y_i^{\prime}$ as:
\begin{equation}
\label{equ:cyi}
    y_i^{\prime}= \omega_i \boxtimes \mathsf{Enc}(\mathcal{M}_i) \boxplus \mathsf{Enc}(R_i)
\end{equation}

Note that the $\mathsf{Enc}(R_i)$ should be uploaded before the model aggregation step according to the analysis in Section \ref{sec:technical-intuition} to prevent the user from manipulating $R_i$ for their benefits. Besides, since $\omega_i$ is a public value and $\mathsf{Enc}(\mathcal{M}_i)$ has been uploaded during data disparity evaluation, the server can calculate $y_i^{\prime}$ using Pallier HE operations.

The detailed description of our secure weighted aggregation scheme for one FL round is given in Protocol \ref{fig:our_protocol}. Note that the \textbf{Setup} step is only performed once at the very beginning of federated learning. Step $0$ to step $4$ are performed in every FL round. 

\subsection{Discussion on Privacy and Security}
\label{sec:disc-sec}

The security of our proposed scheme can be derived from the reasonable assumptions for the FL system as described in Section \ref{sec:technical-intuition}, and the security of underlay cryptographic building blocks. Recall that the requirements of the proposed scheme are to ensure the privacy of the client's local dataset and model, and security against fraudulent messages (weight, model, and view of dropped user list). 

First, we give the discussion from the angle of privacy protection. Since both the client's dataset $\mathcal{D}_i$ and locally trained model $\mathcal{M}_i$ are encrypted using Paillier crypto-system in step Setup and step 1, then uploaded to the server, the server learns nothing from the ciphertext. After that, during computing the weight in step 1, HE operations are applied, and thus additions between ciphertexts do not leak any information. For computing a polynomial $f(l)$ which consists of multiplication between ciphertexts where the server keeps $ \text{Enc}(l)$, a trick is adopted as described in Section \ref{sec:data-disparity}. The main idea is to use a random $r$ to mask $l$ to enable the computation $f(l+r)$ over plaintext, then deduct the redundant terms of $f(l+r)$ to obtain the encryption of $f(l)$ using HE. The security of the trick is based on that the client does not know the value of $r$. Note that this trick enables the client to have the ciphertext of both $r$ and the $ar$ where $a$ is a public value. Privacy is still guaranteed thanks to the IND-CPA property (see Section \ref{sec:Paillier} and Appendix \ref{app:ind-cpa}) of Paillier crypto-system. The privacy of the public masked weighted model $w_i\mathcal{M}_i+R_i$ in step 3 and step 4 is protected by the mask $R_i$ randomly chosen by the client $u_i$, and the secure aggregation scheme. Therefore, there is no information leakage of both client's local dataset and model during the execution of our proposed protocol.

Next, we discuss from the angle of security against fraudulent messages. According to the assumptions described in Section \ref{sec:technical-intuition}, the server has the encryption of the real client's model $y_i$, real mutual cross-entropy $E_i$, and real mask $R_i$ after step setup, step 0, and step 1. Since zero-knowledge proof is adopted in step 2 to ensure the correct decryption of $E_i$, both the server and client can compute real $w_i$. Then the server can compute the encryption of real masked weighted model $w_i\mathcal{M}_i+R_i$, which can be used to ensure that the client uploads the plaintext of the real masked weighted model using zero-knowledge proof in step 3. The security regarding the view of dropped user list is provided by the consistency check technique (see Appendix \ref{app:con-check}) used in the secure aggregation scheme. Therefore, our scheme guarantees security against fraudulent messages, i.e., the correctness of both weight, model, and dropped user list.


\begin{figure}
\begin{protocol}
\begin{adjustbox}{minipage=0.94\linewidth,fbox={\fboxrule} 11pt}
\vspace{-0.2cm}
\centerline{\fontsize{10pt}{12pt}\selectfont{\bfseries{Protocol: Secure Weighted Aggregation \par}}}
\vspace{0.1cm}
\textbf{{Private inputs:}} Each user $u_i$ has a dataset $\mathcal{D}_i=\{(x_u,y_u)|u=1,2,\dots,n_i\}$, a locally trained model $\mathcal{M}_i$, and a Paillier key pair $(sk_i,pk_i)$. Sever has a benchmark dataset $\mathcal{D}_s=\{(x_s,y_s)|u=1,2,\dots,n_s\}$.\\
\textbf{{Public inputs:}} Each user $u_i$'s Paillier public key $pk_i$. The global FL model $\mathcal{M}$, the selected user set $\mathcal{U'}\subseteq \mathcal{U}$, and a threshold $t$ for Shamir secret sharing scheme.\\
\textbf{{Outputs:}} The weighted aggregation of all local models from the alive users $\sum_{i}w_i\mathcal{M}_i$.\\
\hrule
\vspace{-0.1cm}
\begin{itemize}
    \item \textbf{{Setup - Uploading Encrypted Dataset (performed only once):}}\\
    User $u_i$:
    \begin{enumerate}
    \item Uses $pk_i$ to encrypt its dataset $\mathcal{D}_i$ as $\mathsf{Enc}(\mathcal{D}_i)=\{(\mathsf{Enc}(x_u),\mathsf{Enc}(y_u))|u=1,2,\dots,n_i\}$. Then sends $\mathsf{Enc}(\mathcal{D}_i)$ to the server.
    \end{enumerate}
    Server: 
    \begin{enumerate}
    \item Receives $\mathsf{Enc}(\mathcal{D}_i)$ from user $u_i\in \mathcal{U}$.
    \end{enumerate}

    \item \textbf{{Step 0 - Initialization (Init) :}}\\
    User $u_i$:
    \begin{enumerate}
        \item Given the user set $\mathcal{U}$ and the threshold $t$, computes the matrix $R_i$ according to $\pi_{MG}$, such that $(\{R_i\}_{u_i\in \mathcal{U}_1}, \mathcal{U}_1) \leftarrow \pi_{MG}(\mathcal{U}, t)$, where the output $\mathcal{U}_1$ is the set of alive users. Then ecrypts $R_i$ and send $\mathsf{Enc}(R_i)$ to the server.
    \end{enumerate}
    Server:
    \begin{enumerate}
        \item Collects $\mathsf{Enc}(R_i)$ from the set of alive users denoted as $\mathcal{U}_2\subseteq \mathcal{U}_1$, and publishes $\mathcal{U}_2$.
    \end{enumerate}
    \item \textbf{{Step 1 - Computing Mutual Cross-entropy (CompE) :}}\\
    User $u_i$:
    \begin{enumerate}
    \item Uses local dataset $\mathcal{D}_i$ to train local model $\mathcal{M}_i$, and sends $\mathsf{Enc}(\mathcal{M}_i)$ to the server.
    \end{enumerate}
    Server:
    \begin{enumerate}
    \item Receives all $\mathsf{Enc}(\mathcal{M}_i)$ from the set of alive users denoted as $\mathcal{U}_3 \subseteq \mathcal{U}_2$.
    \item For each user $u_i\in \mathcal{U}_3$, computes the mutual cross-entropy $\mathsf{Enc}(E_i)$ of the user $u_i$ based on Equation~\ref{equ:entropy},~\ref{equ:entropy_LS} and \ref{equ:entropy_LL}, and then sents it to user $u_i$.
    \end{enumerate}

    \item \textbf{{Step 2 - Proof of Knowledge of Mutual Cross-entropy (PoKE) :}}
    \begin{enumerate}
        \item \textcolor{red}{\underline{The user $u_i$ decrypts $\mathsf{Enc}(E_i)$ to get $E_i^{\prime}$ and sends it to the server.}}
        \item The user $u_i$ and the server jointly follow Algorithm \ref{alg:p-zkp-m} to get $\beta_E \leftarrow \mathsf{PPoPK}(\mathsf{Enc}(E_i),E_i^{\prime},$ $pk_i,sk_i)$. 
        \item If $\beta_E = 0$, the server removes $u_i$ from $\mathcal{U}_3$. Denote the rest of alive users as $\mathcal{U}_4$. Else both the user $u_i$ and the server compute $\omega_i$ in Equation \ref{equ:omega}.
    \end{enumerate}
    \item \textbf{{Step 3 - Proof of Knowledge of Masked Weighted Model (PoKM) :}}\\
    User $u_i$:
    \begin{enumerate}
    \item \textcolor{red}{\underline{Uploads $y_i^{\prime}=w_i\mathcal{M}_i+R_i$ to the server.}}
    \end{enumerate}
    Server:
    \begin{enumerate}
        \item Receives $y_i^{\prime}$ from the set of current alive users denoted as $\mathcal{U}_5$. For each user $u_i\in \mathcal{U}_5$, computes $y_i= \omega_i \boxtimes \mathsf{Enc}(\mathcal{M}_i) \boxplus \mathsf{Enc}(R_i)$ according to Equation \ref{equ:cyi}.
        \item Cooperate with user $u_i$ following Algorithm \ref{alg:p-zkp-m} to get $\beta_M \leftarrow \mathsf{PPoPK}(y_i,y_i^{\prime},$ $pk_i,sk_i)$
        \item If $\beta_M = 0$, the server removes $u_i$ from $\mathcal{U}_5$, denote the rest of alive users as $\mathcal{U}_6$.
        \item Computes the aggregation of masked weighted models such that $y=\sum_{u_i\in \mathcal{U}_6}y_{u_i}$ following $\pi_{MMA}$ protocol.
    \end{enumerate}
    
    \item \textbf{{Step 4 - Weighted Aggregation (WAgg) :}}\\
    Server: 
    \begin{enumerate}
        \item \textcolor{red}{\underline{Publishes the alive user set $\mathcal{U}_6$}}, and cooperates with clients $u_i\in \mathcal{U}_6$ to jointly compute and output the final weighted aggregation $\mathcal{M}$ according to the protocol $\pi_{MAR}$, such that $\mathcal{M} \leftarrow \pi_{MAR}(y, \mathcal{U}_1,\mathcal{U}_6,t)$.
    \end{enumerate}

\end{itemize}
\end{adjustbox}
\caption{Detailed description of proposed aggregation protocol. \textcolor{red}{\underline{The participants may }} \textcolor{red}{\underline{deviate from the protocol by sending fraudulent messages in the red underlined operations.}}}
\label{fig:our_protocol}
\end{protocol}
\end{figure}

\section{Performance Analysis}
\label{sec:performance}
\textbf{System setting.} Our prototype is tested over two Linux workstations with an Intel Xeon E5-2603v4 CPU (1.70GHz) and 64 GB of RAM, running CentOS 7 in the same region. The average latency of the network is 0.207 ms, and the average bandwidth is 1.25 GB/s. The central server runs in multi-thread on one workstation, while the users are running in parallel in the other workstation. In our experiments, original data is represented in fixed-point with 44 bits, including 17 bits for the integer part and 27 bits for the fractional part. The ring size $N$ for Paillier is set to be a 1024-bit integer, and the security parameter $\kappa$ is set to be $80$. We use AES-GCM with 128-bit keys for authenticated channels, and adopt the state-of-the-art secure aggregation scheme (Secagg+ \cite{bell2020secure}) for dropout-resilient aggregation. The performance is evaluated on Credit Card Clients dataset from the UCI ML repository \cite{dua2017uci}, from which each user randomly chooses samples with the same size to construct its local dataset. The number of samples in the benchmark dataset on the server-side is set to be 500.

As seen in Figure \ref{fig:ppopk_run-time} and \ref{fig:ppopk_communication_cost}, the run time and communication cost for $\mathsf{PPoPK}$ increase almost linearly with the dimensions of input $m$. Figure \ref{fig:dde_run-time} and \ref{fig:dde_communication_cost} show the run time and communication cost of a client in the first phase, i.e., data disparity evaluation, including CompE and PoKE. Both the run time and communication costs increase linearly with the number of samples in the local dataset. Figure \ref{fig:total_run-time} and \ref{fig:total_communication_cost} give the performance on total run time and communication cost of the second phase, i.e., weighted aggregation, including PoKM and Wagg, compared with the model aggregation counterpart in secure aggregation scheme (Secagg) \cite{bonawitz2017practical} and the improved secure aggregation scheme (Secagg+) \cite{bell2020secure}. Note that the dropped users in our scheme include both the users that do not pass PoKM and the users that drop out of the system. The results show that our scheme provides an additional security guarantee against fraudulent messages with affordable overheads compared to the previous works. In particular, our scheme achieves around 1.2 times in run time and 1.3 times in communication cost upon the state-of-the-art secure aggregation scheme Secagg+, which indicates the practicality of our scheme. Table \ref{tab:total_table} lists the run time and communication costs of the user and server for each step in one FL round with different dropout rates in the first phase for data disparity evaluation $R_1$, and the second phase for weighted aggregation $R_2$. Note that the server's run time is the run time of the whole FL system, and we treat the users who do not pass PoKE and PoKM as dropout users as demonstrated in Protocol \ref{fig:our_protocol}. In specific, we consider the "worst-case" dropout scenario where the users drop out during PoKE in the first phase as the server has already completed the calculation of the weights, or PoKM in the second phase as the users have already uploaded their weighted models. We can observe that dropout users cause a significant increasement in the total run time because of the high cost of dealing with dropout users in the underlying secure aggregation scheme.

\vspace{-0.8cm}
\begin{figure}[H]
\centering
\addtolength{\subfigcapskip}{-0.13in}
\subfigure[Run time evaluation]{
\begin{minipage}[t]{0.45\linewidth}
\centering
\includegraphics[width=1\linewidth]{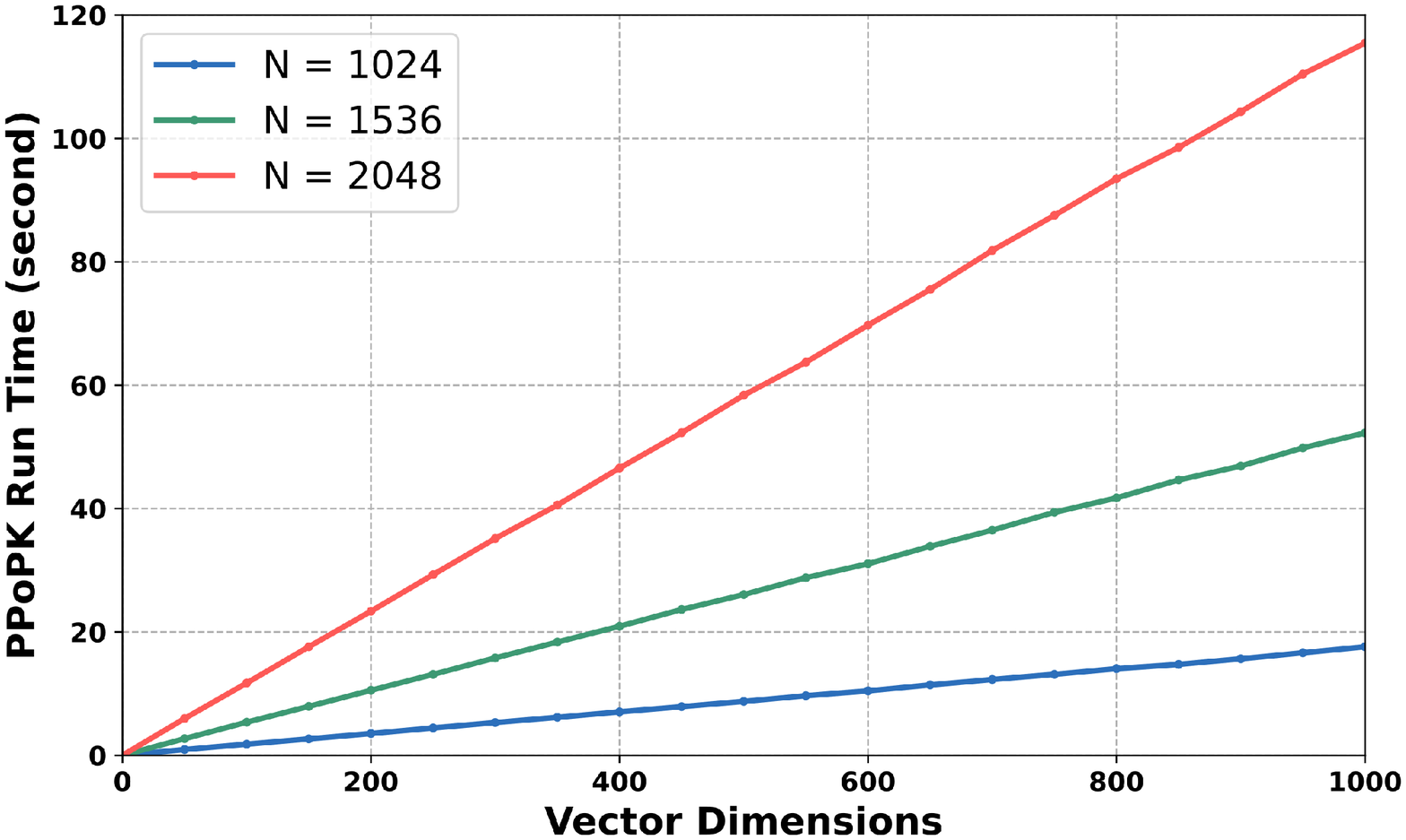}
\label{fig:ppopk_run-time}
\end{minipage}
}
\subfigure[Communication cost evaluation]{
\begin{minipage}[t]{0.45\linewidth}
\centering
\includegraphics[width=1\linewidth]{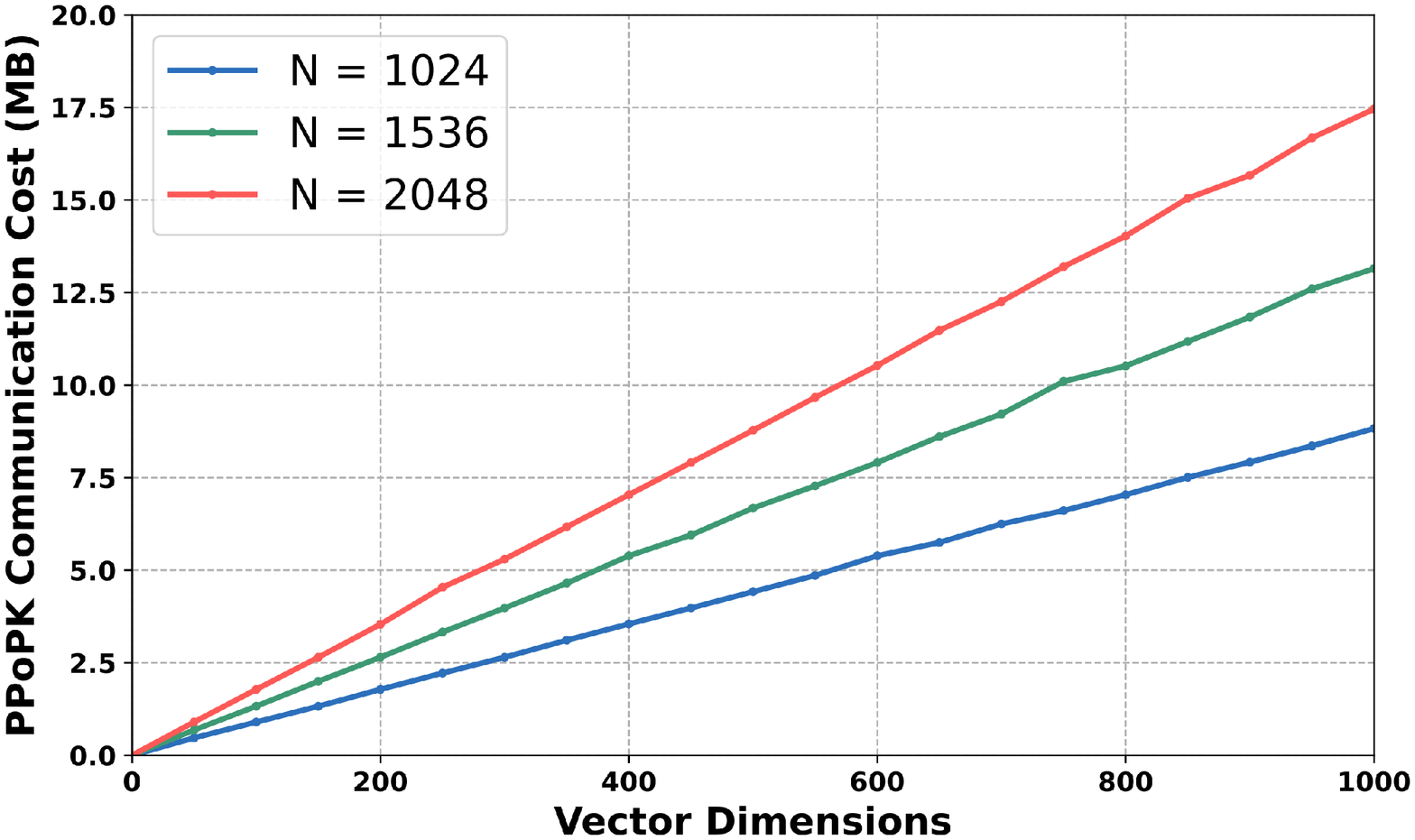}
\label{fig:ppopk_communication_cost}
\end{minipage}
}
\centering
\vspace{-0.3cm}
\caption{The run time and communication cost of algorithm $\mathsf{PPoPK}(c,m,$ $pk_i,sk_i)$ for the user $u_i$ and the server $S$ as the dimension of input vector $m$ increases with different ring size $N$ for Paillier crypto-system.}
\end{figure}

\vspace{-2cm}
\begin{figure}[H]
\centering
\addtolength{\subfigcapskip}{-0.13in}
\subfigure[Run time evaluation]{
\begin{minipage}[t]{0.45\linewidth}
\centering
\includegraphics[width=1\linewidth]{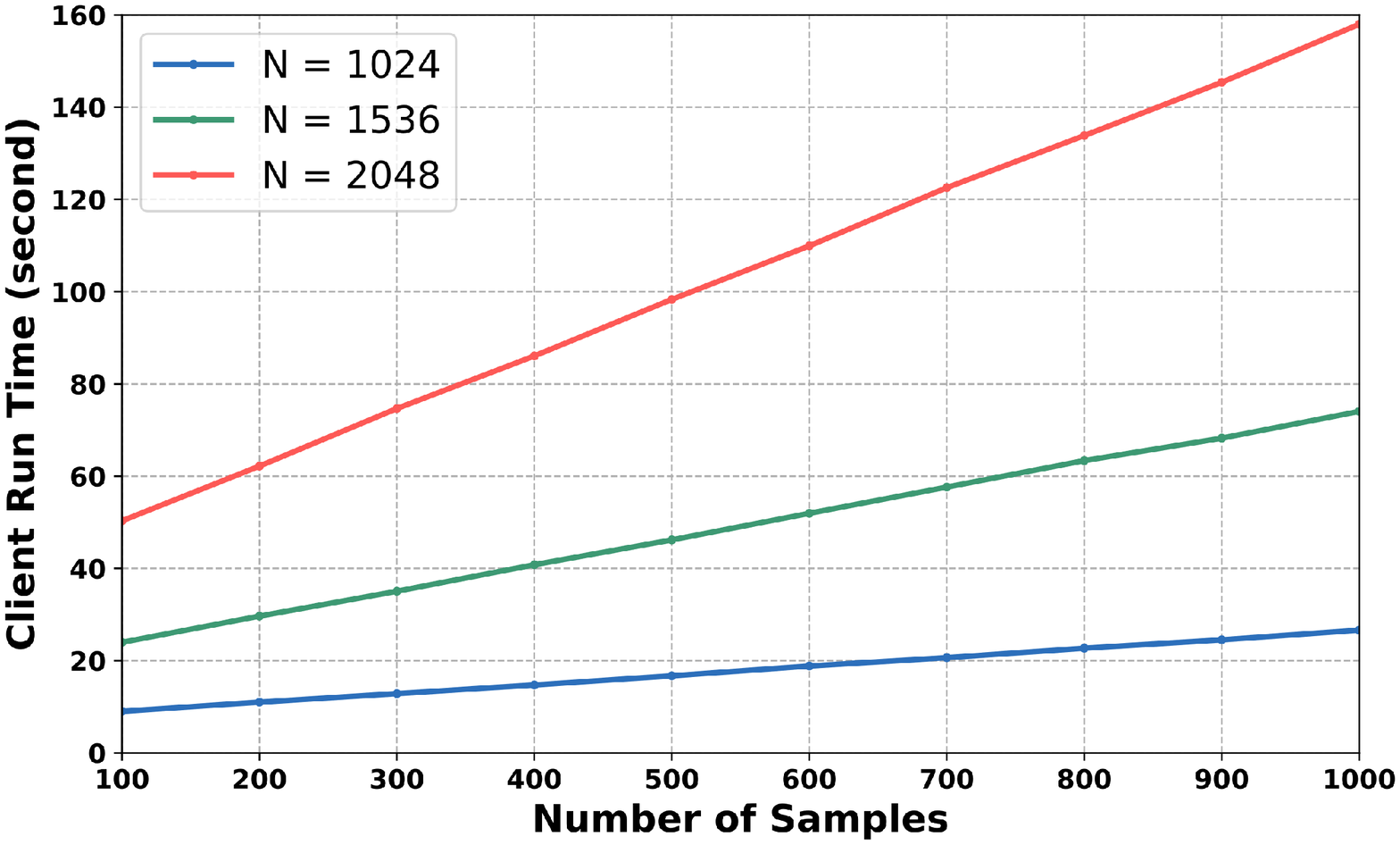}
\label{fig:dde_run-time}
\end{minipage}%
}
\subfigure[Communication cost evaluation]{
\begin{minipage}[t]{0.45\linewidth}
\centering
\includegraphics[width=1\linewidth]{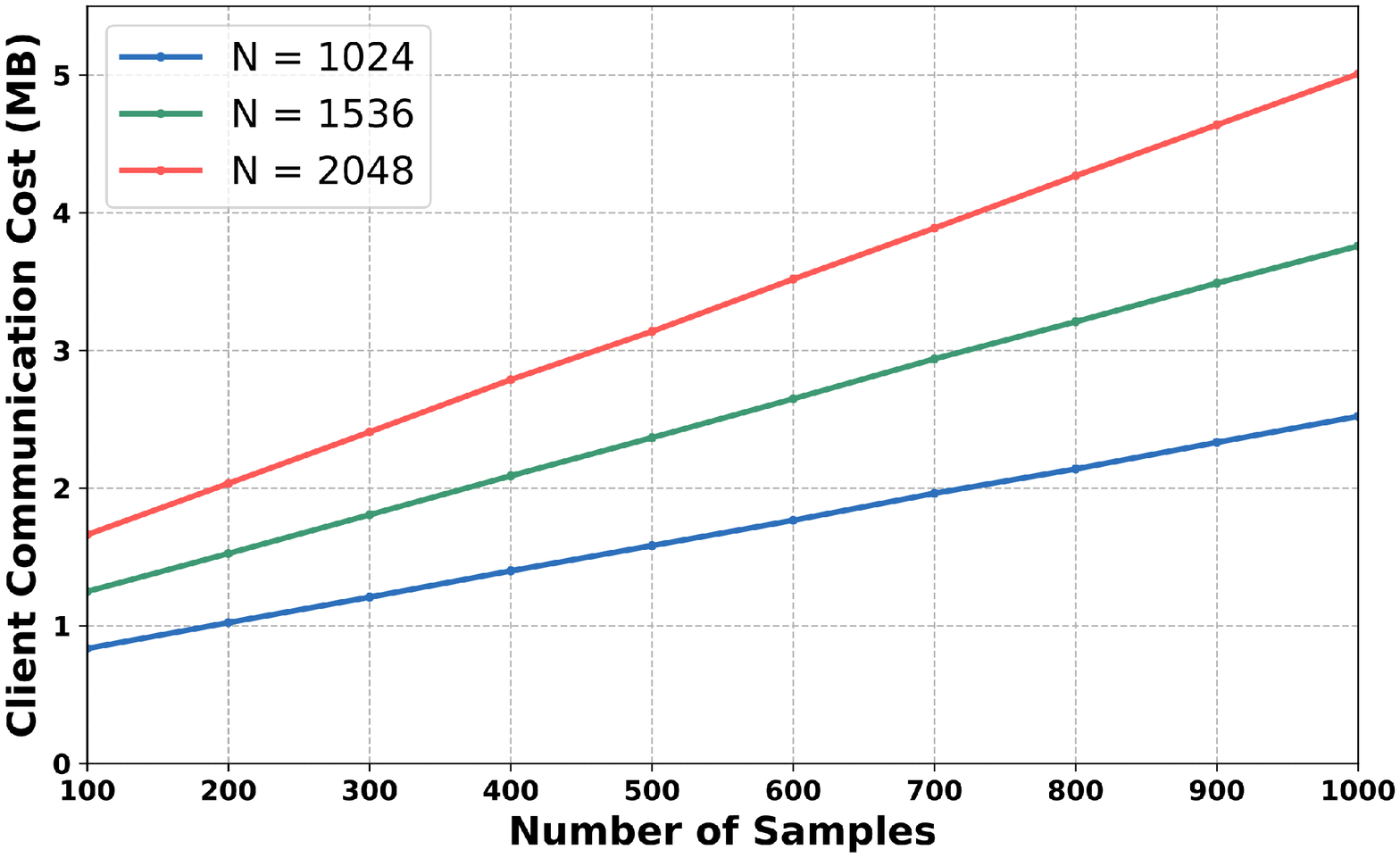}
\label{fig:dde_communication_cost}
\end{minipage}%
}%
\vspace{-0.1cm}
\centering
\caption{The run time and communication cost per client for data disparity evaluation as the number of samples in local dataset increases with different ring size $N$ for Paillier crypto-system.}
\centering
\subfigure[Run time evaluation]{
\begin{minipage}[t]{0.45\linewidth}
\centering
\includegraphics[width=1\linewidth]{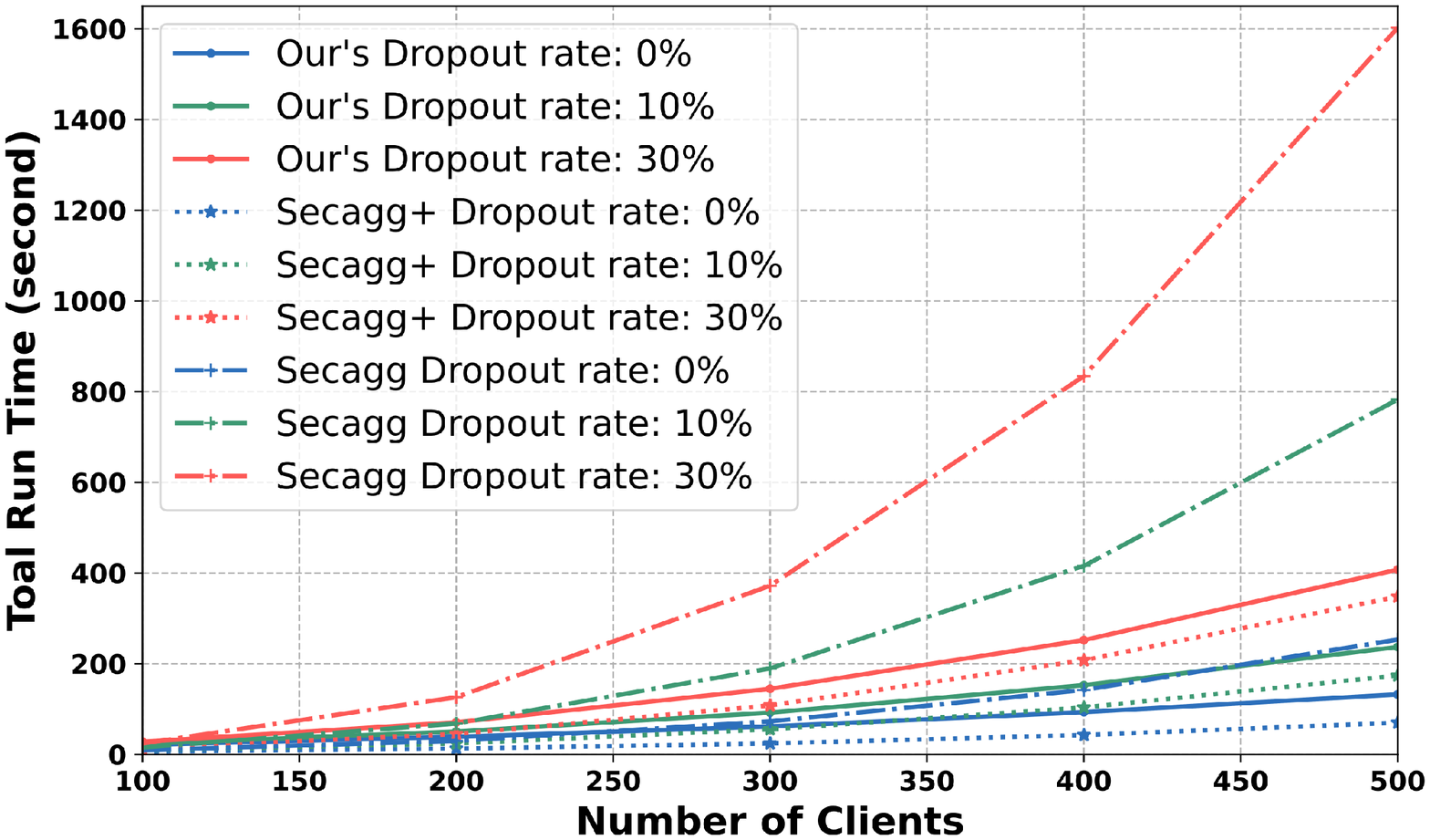}
\label{fig:total_run-time}
\end{minipage}%
}
\subfigure[Communication cost evaluation]{
\begin{minipage}[t]{0.45\linewidth}
\centering
\includegraphics[width=1\linewidth]{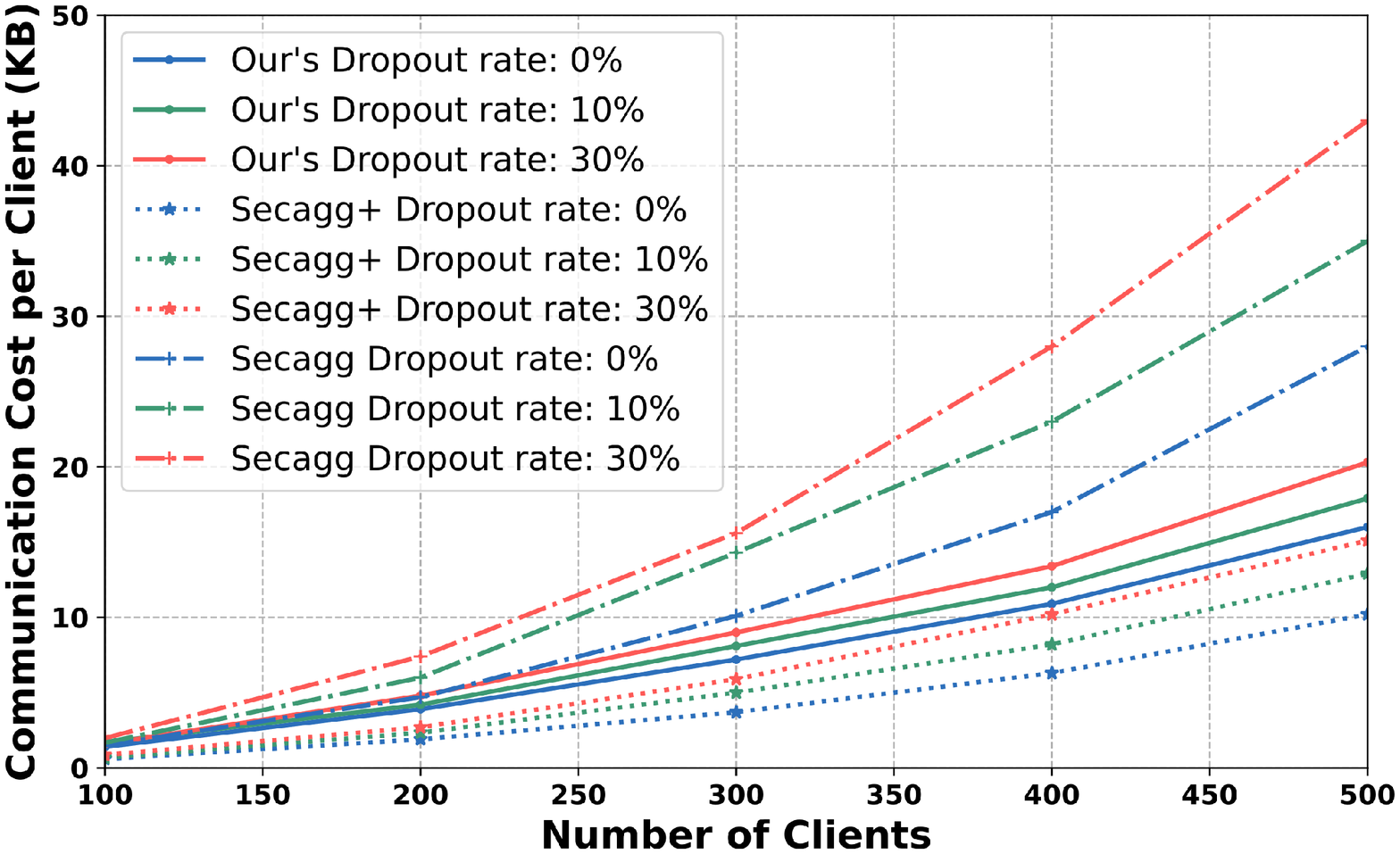}
\label{fig:total_communication_cost}
\end{minipage}%
}%
\vspace{-0.1cm}
\centering
\caption{The total run time and communication cost of weighted aggregation including PoKM and WAgg for one FL round, compared with standard Secagg and Secagg+ scheme.}
\end{figure}

\vspace{-0.5cm}

\begin{table}[htbp]
\centering
\caption{The run time (s)/communication costs (MB) of the user and server for each step in one FL round. The number of clients is 500. $R_1$ and $R_2$ are the fractions of dropout users in the first phase and the second phase respectively.}
\begin{tabular}{|c|c|c|c|c|c|c|c|c|}
\hline
\multirow{2}{*}{} & \multirow{2}{*}{$R_1$} & \multirow{2}{*}{$R_2$} & \multicolumn{1}{c|}{\multirow{2}{*}{Init}} & \multicolumn{2}{c|}{The First Phase}            & \multicolumn{2}{c|}{The Second Phase}                  & \multicolumn{1}{c|}{\multirow{2}{*}{Total}} \\ \cline{5-8}
                  &                        &                        & \multicolumn{1}{c|}{}                      & \multicolumn{1}{c|}{ComE} & \multicolumn{1}{c|}{PoKE} & \multicolumn{1}{c|}{PoKM} & \multicolumn{1}{c|}{WAgg} & \multicolumn{1}{c|}{}                       \\ \hline
user              & 0                      & 0                      & 0.932/0.48                                   & 5.648/1.02                 & 0.008/0.01                    & 0.023/0.01                   & 0.039/0.02                   & 6.650/1.54                                   \\ \hline
~server~            & 0                      & 0                      & 1.223/240                                   & 10.972/512                 & 0.043/3                      & 0.130/6                     & 67.310/10                  & 79.678/771                                   \\ \hline
server            & 10\%                   & 0                      & ~1.207/240~                                   & ~10.917/513~                 & 0.036/3                      & 0.118/6                     & ~129.785/11~                 & ~142.063/773~                                  \\
server            & 0                      & 10\%                   & 1.215/240                                   & 10.973/513                 & 0.037/4                      & 0.127/5                     & 130.583/12                 & 142.935/774                                  \\
server            & 5\%                    & 5\%                    & 1.211/240                                   & 10.938/513                 & 0.038/3                      & 0.125/5                     & 132.314/13                 & 144.626/774                                  \\ \hline
server            & 30\%                   & 0                      & 1.202/241                                   & 10.908/512                 & 0.034/3                      & 0.095/4                      & 254.237/14                 & 266.476/774                                  \\
server            & 0                      & 30\%                   & 1.204/240                                   & 10.927/512                 & 0.037/4                      & 0.120/4                     & 255.142/15                 & 267.430/775                                  \\
server            & 15\%                   & 15\%                   & 1.204/240                                   & 10.912/513                 & 0.038/3                      & 0.105/3                     & 256.895/15                 & 269.154/775                                  \\ \hline
\end{tabular}
\label{tab:total_table}
\end{table}


\section{Discussion on Related Works}
\label{sec:related-works}
There have been several schemes proposed to perform privacy-preserving and dropout-resilient aggregation. The main idea of these schemes is to integrate FL with privacy-preserving techniques such as DP, HE, and MPC. For example, the authors in \cite{mcmahan2018learning} propose a DP-based FL scheme that a level of noise is added to each client's locally trained model before aggregation, and thus to hide clients' data distribution. However, there is a trade-off between privacy loss and model performance. 

For the HE-based aggregation scheme, in every round of FL, each client's locally trained model is encrypted before uploading. Then the central server aggregates the clients' model in ciphertext using HE operations and publishes the result for decryption. In this scheme, the private key is distributed among all participants relying on a threshold HE crypto-system, e.g., \cite{rastogi2010differentially}, and several parties that more than a threshold number can jointly decrypt an encrypted message.
Although the HE-based scheme provides an exact solution for privacy-preserving aggregation compared with the DP-based scheme, it may incur infeasible cost overhead for a large-scale FL system as the underlying threshold crypto-system requires all participants to contribute to several expensive protocols. Besides, this scheme cannot deal with clients dropping out. 

A more efficient and dropout-robust aggregation scheme providing an exact solution is based on MPC, as proposed in \cite{bonawitz2017practical}. In this scheme, each client's locally trained model is masked that the seeds for generating masks are secretly shared among all clients using a threshold secret sharing scheme to handle the dropout users. The improved version \cite{bell2020secure} upon \cite{bonawitz2017practical} replaces the complete communication graph in \cite{bonawitz2017practical} with a $k$-regular graph to further reduce the communication overheads. Another variant TurboAgg \cite{so2021turbo} divides clients into multiple groups and follows a multi-group circular structure for communication. FastSecAgg \cite{kadhe2020fastsecagg} substitutes threshold secret sharing with FFT-based multi-secret sharing scheme. NIKE \cite{mandal2018nike} adopts a non-interactive key exchange protocol relying on non-colluding servers. 

Note that for a large-scale FL system, executions of threshold secret sharing operations are more cost-effective than that of a threshold HE crypto-system. Nevertheless, all of the schemes mentioned above focus on achieving FL from a privacy-preserving perspective, issues still exist that both the central server and clients may send fraudulent messages to each other for their benefits.

\section{Conclusion}
\label{sec:conclusion}
In this paper, we proposed a privacy-enhanced FL scheme for supporting secure weighted aggregation. Our scheme is able to deal with both data disparity, data privacy, and dishonest participants (who send fraudulent messages to manipulate the computed weights) in FL systems. The proposed scheme leverages on HE and zero-knowledge proof (ZKP), and combines with a dropout-resilient secure aggregation scheme. Experimental results show that our scheme is practical and secure. Compared to existing FL approaches, our scheme achieves secure weighted aggregation with an additional security guarantee against fraudulent messages with affordable runtime overheads and communication costs.

For future research investigation, we will substitute the Paillier cryptosystem with other HE schemes such as CKKS \cite{cheon2017homomorphic} in order to improve the efficiency of data disparity evaluation. Several other attacks such as membership inference attack \cite{shokri2017membership} and model inversion attack \cite{fredrikson2015model} will be taken into account for a more general threat model. Besides, other distributed machine learning schemes such as split learning \cite{gupta2018distributed} can be integrated to further reduce the computation overheads on the client-side to make FL more friendly to resource-constrained edge devices.

\bibliographystyle{splncs04}
\bibliography{mybibliography}

\appendix
\section{Proof of Algorithm \ref{alg:p-zkpopk}}
\label{app:p-zkpopk}

The correctness is proved as follows. If the message $m$ provided by the user is the decryption of $c$, $c'=c-\mathsf{HE.Enc}(pk_i,m,r_s)(\bmod n^2)$ should be the ciphertext of zero with respect to $r'$ that $c'\bmod n^2=(c')^{n\cdot d} \bmod n^2=(r')^n\bmod n^2= \mathsf{HE.Enc}(pk_i,0,r')$. Therefore, the output $\beta$ of the algorithm $\mathsf{PZKPoPKoZ}$ in step 3 is 1. Similarly, if $m$ is not the decryption of $c$, the algorithm $\mathsf{PZKPoPKoZ}$ outputs $\beta=0$. 

\begin{theorem}[Security against malicious users]
\label{thm:sec}
The Paillier-based Proof of Plaintext Knowledge is secure against malicious users. If there exists a probabilistic polynomial time (PPT) adversary $\mathcal{A}$ with Paillier key pair $(sk,pk)$ providing a fraudulent plaintext $m'$ as the decryption of $c$ that $\mathsf{HE.Dec}(sk,c)\neq m'$. The probability of the server accepting the fraudulent decryption $\mathsf{Pr}[1\leftarrow \mathsf{PPoPK}_{\mathcal{A}}(c,m',pk,sk)] \leqslant 2^{-k/2}$ where $k$ is the bit length of $n$ as the Paillier public key.
\end{theorem}
\begin{proof}
Assume that a PPT adversary $\mathcal{A}$ with key pair $(sk,pk)$ provides a fraudulent decryption $m'$ of a ciphertext $c$, while $\mathsf{HE.Dec}$ $(sk,c)$ $=m$ and $m\neq m'$, the server accepts $m'$ only if $1\leftarrow \mathsf{PPoPK}(c,$ $m',pk,sk)$. Given the algorithm \ref{alg:p-zkp-m}, the server randomly chooses $r_s\in \mathbb{Z}_n^*$ to compute $c'=c-\mathsf{HE.Enc}(pk_i,m',r_s)(\bmod\ n^2) = \mathsf{Enc}(m-m')$ and sends $c'$ to the adversary $\mathcal{A}$. To make the server accept $m'$, the adversary $\mathcal{A}$ has to prove that $c'$ is a ciphertext of zero. Therefore, the security of $\mathsf{PPoPK}$ follows from the security of $\mathsf{PZKPoPKoZ}$ such that the probability of the adversary make the server accept the fraudulent decryption $m'$ is less than $2^{-k/2}$ where $k$ is the bit length of $n$ as the Paillier public key. 
\end{proof}

\section{Indistinguishability under chosen-plaintext attack (IND-CPA)}
\label{app:ind-cpa}

\begin{defx}[IND-CPA]
Let $\mathsf{HE}=(\mathsf{KeyGen},$$ \mathsf{Enc},$$ \mathsf{Dec})$ be a homomorphic encryption scheme, and $\mathcal{A}$ be a probabilistic polynomial-time (PPT) adversary. Consider the following experiment:\\
$\mathbf{EXP}_{\mathsf{HE},\mathcal{A}}^{\mathsf{IND-CPA}}(\lambda)$:
\begin{itemize}
    \item[(1)] $b \stackrel{\$}{\gets} \{0,1\}$, $\{pk,sk\} \stackrel{\$}{\gets} \mathsf{HE.KeyGen}(p,q)$
    \item[(2)] $(m_0,m_1) \gets \mathcal{A}(pk)$
    \item[(3)] $c \gets \mathsf{HE.Enc}(pk,m_b,r)$
    \item[(4)] $b'\gets \mathcal{A}(c)$
    \item[(5)] Output $1$ if $b=b'$, 0 o/w.
\end{itemize}
Define the advantage of the adversary $\mathcal{A}$ as 
$$\mathbf{Adv}_{\mathsf{HE},\mathcal{A}}^{\mathsf{IND-CPA}}(\lambda):=| \mathsf{Pr}[\mathbf{EXP}_{\mathsf{HE},\mathcal{A}}^{\mathsf{IND-CPA}}(\lambda)=1]-\frac{1}{2}| $$
Then we say such HE scheme is indistinguishable under chosen-plaintext attack if for any PPT adversariy $\mathcal{A}$, there exists a negligible function $\mathsf{negl}(\cdot)$ such that $\mathcal{A}$ as $\mathbf{Adv}_{\mathsf{HE},\mathcal{A}}^{\mathsf{IND-CPA}}(\lambda) \leqslant \mathsf{negl}(\lambda)$.
\end{defx}

Paillier crypto-system provides semantic security against chosen-plaintext attacks (IND-CPA). The challenge of successfully distinguishing ciphertext essentially can be deducted to deciding composite residuosity, which is believed to be intractable. Thus, an adversary will be unable to distinguish pairs of ciphertexts based on the message they encrypt using Paillier crypto-system \cite{paillier1999public}.

\section{Consistency Check}
\label{app:con-check}

Consistency check guarantees that a set of users of which the number is greater than the threshold $t$ for the Shamir scheme have the same view of dropped users (see \cite{bonawitz2017practical,bell2020secure} for the detailed security proof). To achieve this goal, consistency check heavily leverage signature schemes which can be used to prove the origin of a message. Denote a signature scheme as $(\mathcal{F}_{kg}, \mathcal{F}_{sig}, \mathcal{F}_{vrfy})$. Randomised algorithm $\mathcal{F}_{kg}\rightarrow (sk, pk)$ outputs a secret key $sk$ and a public key $pk$. Randomised algorithm $\mathcal{F}_{sig}(m,sk)\rightarrow\delta$ uses $sk$ to assign a signature $\delta$ on the message $m$. Deterministic algorithm $\mathcal{F}_{vrfy}(m,pk,\delta)\rightarrow\{0,1\}$ takes in the $pk$, $m$, $\delta$ and returns $1$ if $\delta$ is a signature on $m$ and $0$ otherwise. Assume each alive user $u \in \mathcal{U}_a$ has a key pair $(sk_u, pk_u)$, and the server has a key pair $(sk_s, pk_s)$ generated using $\mathcal{F}_{kg}$. The server maintains the list of alive users as $\mathcal{U}_l$. Consistency check works as follows:
\begin{enumerate}
    \item Each user $u \in \mathcal{U}_a$ fetches the list of $\mathcal{U}_l$ from the server, denoted as $\mathcal{U}_l^u$.
    \item Each user $u \in \mathcal{U}_a$ assigns the signature on received $\mathcal{U}_l^u$ as $\mathcal{F}_{sig}(\mathcal{U}_l^u,sk_u)\rightarrow\delta_u$, and send $(\mathcal{U}_l^u,\delta_u)$ to the server.
    \item The server receives $(\mathcal{U}_l^u,\delta_u)$ from at least $t$ users to build a list denoted as $\mathcal{U}_l^{\prime}$, and forwards them to the user $u \in \mathcal{U}_a$.
    \item For each user $u \in \mathcal{U}_a$, if $|\mathcal{U}_l^{\prime}|<t$ or $\mathcal{U}_l^{\prime} \not\subset \mathcal{U}_l$, abort. 
\end{enumerate}

Note that consistency check provides the privacy guarantee against active adversaries only when the threshold $t$ is set to be a proper value. We discuss its bounds in Appendix \ref{app:disc-tol}.

\section{Discussion on Tolerance}
\label{app:disc-tol}
Next, we discuss the maximum number of malicious adversaries allowed and the lower bound of the threshold $t$ for the Shamir scheme. Recall that there are $n$ clients including $n_a$ malicious clients. Let the number of signatures from honest clients on two list $l_1$, $l_2$ be $n_1$ and $n_2$ respectively, and the number of signatures from malicious clients is $n_a$. To have both $l_1$ and $l_2$ to pass the consistency check, the number of signatures on them should be greater than the threshold $t$ such that $n_1+n_a \geq t, n_2+n_a \geq t$, subject to $ n_1+n_2+n_a\leq n$. To guarantee the security, we have $t\geq\frac{n+n_a}{2}$. Beside, since the server cannot the forge signatures of the honest parties, there is no two sets that $2(t-n_a)\leq n-n_a$ which implies the maximum number of malicious adversaries allowed is $\lceil \frac{n}{3} \rceil-1$, and thus the lower bound threshold $t$ is $\lfloor \frac{2n}{3} \rfloor+1$. Since $\lceil \frac{n}{3} \rceil-1 < \lfloor \frac{2n}{3} \rfloor+1 \leq t$, the security of Shamir secret sharing scheme is also guaranteed. More details about the derivation of the above-mentioned bounds can be found in \cite{bonawitz2017practical}. Note that dropped users do not affect the privacy and security of step setup and step 0-3, as each client executes the protocol with the server independently. Therefore, the bounds of the number of malicious adversaries and threshold regarding tolerance still follow those of the secure aggregation scheme.

%% file: samplepaper.bbl
\begin{thebibliography}{10}
\providecommand{\url}[1]{\texttt{#1}}
\providecommand{\urlprefix}{URL }
\providecommand{\doi}[1]{https://doi.org/#1}

\bibitem{bell2020secure}
Bell, J.H., Bonawitz, K.A., Gasc{\'o}n, A., Lepoint, T., Raykova, M.: Secure
  single-server aggregation with (poly) logarithmic overhead. In: Proceedings
  of the 2020 ACM SIGSAC Conference on Computer and Communications Security.
  pp. 1253--1269 (2020)

\bibitem{blanchard2017machine}
Blanchard, P., El~Mhamdi, E.M., Guerraoui, R., Stainer, J.: Machine learning
  with adversaries: Byzantine tolerant gradient descent. In: Proceedings of the
  31st International Conference on Neural Information Processing Systems. pp.
  118--128 (2017)

\bibitem{blum1984generate}
Blum, M., Micali, S.: How to generate cryptographically strong sequences of
  pseudorandom bits. SIAM journal on Computing  \textbf{13}(4),  850--864
  (1984)

\bibitem{bonawitz2017practical}
Bonawitz, K., Ivanov, V., Kreuter, B., Marcedone, A., McMahan, H.B., Patel, S.,
  Ramage, D., Segal, A., Seth, K.: Practical secure aggregation for
  privacy-preserving machine learning. In: Proceedings of the 2017 ACM SIGSAC
  Conference on Computer and Communications Security. pp. 1175--1191 (2017)

\bibitem{chen2020focus}
Chen, Y., Yang, X., Qin, X., Yu, H., Chen, B., Shen, Z.: Focus: Dealing with
  label quality disparity in federated learning. International Workshop on
  Federated Learning for User Privacy and Data Confidentiality in Conjunction
  with IJCAI 2020  (2020)

\bibitem{cheon2017homomorphic}
Cheon, J.H., Kim, A., Kim, M., Song, Y.: Homomorphic encryption for arithmetic
  of approximate numbers. In: International Conference on the Theory and
  Application of Cryptology and Information Security. pp. 409--437. Springer
  (2017)

\bibitem{damgaard2010generalization}
Damg{\aa}rd, I., Jurik, M., Nielsen, J.B.: A generalization of paillier’s
  public-key system with applications to electronic voting. International
  Journal of Information Security  \textbf{9}(6),  371--385 (2010)

\bibitem{dua2017uci}
Dua, D., Taniskidou, E.K.: Uci machine learning repository [http://archive.
  ics. uci. edu/ml]. university of california, irvine. School of Information
  and Computer Sciences  (2017)

\bibitem{fredrikson2015model}
Fredrikson, M., Jha, S., Ristenpart, T.: Model inversion attacks that exploit
  confidence information and basic countermeasures. In: Proceedings of the 22nd
  ACM SIGSAC Conference on Computer and Communications Security. pp. 1322--1333
  (2015)

\bibitem{goldreich1994definitions}
Goldreich, O., Oren, Y.: Definitions and properties of zero-knowledge proof
  systems. Journal of Cryptology  \textbf{7}(1),  1--32 (1994)

\bibitem{gupta2018distributed}
Gupta, O., Raskar, R.: Distributed learning of deep neural network over
  multiple agents. Journal of Network and Computer Applications  \textbf{116},
  ~1--8 (2018)

\bibitem{kadhe2020fastsecagg}
Kadhe, S., Rajaraman, N., Koyluoglu, O.O., Ramchandran, K.: Fastsecagg:
  Scalable secure aggregation for privacy-preserving federated learning. arXiv
  preprint arXiv:2009.11248  (2020)

\bibitem{kairouz2019advances}
Kairouz, P., McMahan, H.B., Avent, B., Bellet, A., Bennis, M., Bhagoji, A.N.,
  Bonawitz, K., Charles, Z., Cormode, G., Cummings, R., et~al.: Advances and
  open problems in federated learning. arXiv preprint arXiv:1912.04977  (2019)

\bibitem{mandal2018nike}
Mandal, K., Gong, G., Liu, C.: Nike-based fast privacy-preserving
  highdimensional data aggregation for mobile devices. Tech. rep., CACR
  Technical Report, CACR 2018-10, University of Waterloo, Canada (2018)

\bibitem{mcmahan2017communication}
McMahan, B., Moore, E., Ramage, D., Hampson, S., y~Arcas, B.A.:
  Communication-efficient learning of deep networks from decentralized data.
  In: Artificial Intelligence and Statistics. pp. 1273--1282 (2017)

\bibitem{mcmahan2018learning}
McMahan, H.B., Ramage, D., Talwar, K., Zhang, L.: Learning differentially
  private recurrent language models. In: International Conference on Learning
  Representations (2018)

\bibitem{paillier1999public}
Paillier, P.: Public-key cryptosystems based on composite degree residuosity
  classes. In: International conference on the theory and applications of
  cryptographic techniques. pp. 223--238. Springer (1999)

\bibitem{rastogi2010differentially}
Rastogi, V., Nath, S.: Differentially private aggregation of distributed
  time-series with transformation and encryption. In: Proceedings of the 2010
  ACM SIGMOD International Conference on Management of data. pp. 735--746
  (2010)

\bibitem{shokri2017membership}
Shokri, R., Stronati, M., Song, C., Shmatikov, V.: Membership inference attacks
  against machine learning models. In: 2017 IEEE Symposium on Security and
  Privacy (SP). pp. 3--18. IEEE (2017)

\bibitem{so2021turbo}
So, J., G{\"u}ler, B., Avestimehr, A.S.: Turbo-aggregate: Breaking the
  quadratic aggregation barrier in secure federated learning. IEEE Journal on
  Selected Areas in Information Theory  (2021)

\bibitem{zhao2020local}
Zhao, Y., Zhao, J., Yang, M., Wang, T., Wang, N., Lyu, L., Niyato, D., Lam,
  K.Y.: Local differential privacy based federated learning for internet of
  things. IEEE Internet of Things Journal  (2020)

\bibitem{zhu2019deep}
Zhu, L., Liu, Z., Han, S.: Deep leakage from gradients. In: Advances in Neural
  Information Processing Systems. pp. 14774--14784 (2019)

\end{thebibliography}
